\documentclass[12pt]{article}

\usepackage[margin=1in]{geometry}
\usepackage[T1]{fontenc}
\usepackage[utf8]{inputenc}
\usepackage{mathptmx}
\usepackage{microtype}
\usepackage{amsmath,amssymb,amsthm,bm}
\usepackage{booktabs,tabularx,array}
\usepackage{float}
\usepackage[dvipsnames]{xcolor}
\usepackage{enumitem}
\usepackage{natbib}
\usepackage{csquotes}
\usepackage{setspace}
\usepackage{fancyhdr}
\usepackage{titlesec}
\usepackage{caption}
\usepackage[colorlinks=true,linkcolor=MidnightBlue,citecolor=MidnightBlue,urlcolor=MidnightBlue]{hyperref}

\definecolor{statementblue}{HTML}{1F4E79}
\definecolor{lightblue}{HTML}{EAF1F7}

\newtheorem{definition}{Definition}
\newtheorem{proposition}{Proposition}
\newtheorem{principle}{Principle}

\newcommand{\W}{\textsf{W}}
\newcommand{\E}{\textsf{E}}
\newcommand{\Prb}{\textsf{Pr}}
\newcommand{\Risk}{\textsf{R}}
\newcommand{\statementheading}[2]{%
  \subsection{#1}%
  \noindent{\textbf{Statement.}}\;
  \textit{#2}\par\medskip
}

\titleformat{\section}{\large\bfseries}{\thesection.}{0.55em}{}
\titleformat{\subsection}{\normalsize\bfseries}{\thesubsection.}{0.55em}{}
\titleformat{\subsubsection}{\normalsize\itshape}{\thesubsubsection.}{0.55em}{}

\setlist{itemsep=0.2em,topsep=0.35em}
\title{\Huge \textbf{Statistics in the Age of AI}}

\author{
{\textbf{Juan Sosa}}\\
 Departamento de Estadística, Universidad Nacional de Colombia\\
 Bogotá, Colombia
\and
{\textbf{Brenda Betancourt}}\\
 Department of Statistics, George Mason University\\
Fairfax, Virginia, USA
}

\date{}

\begin{document}

\maketitle

\begin{abstract}
Artificial intelligence (AI) can automate programming, model fitting, visualization, simulation, literature synthesis, and increasingly sophisticated methodological tasks, but it cannot remove the logical conditions under which data support scientific claims or consequential decisions. We formalize these conditions through \emph{statistical warrant}, which connects data to a claim through the target, observation regime, assumptions, procedure, uncertainty assessment, validation criterion, loss structure, governance and accountability. No algorithm can consistently recover a target that is not identified by the observation regime without additional information or assumptions. From this principle, we organize the argument around five statements. Questions and targets are integral to statistical methods. Data acquire evidential meaning only through design, provenance, and assumptions. Description, prediction, causal inference, and decision are mathematically distinct tasks. Analytical abundance requires accounting for how analyses are selected, uncertainty across the analytical system, and deployment validation. The statistician's fundamental role is therefore to construct, criticize, and safeguard statistical warrant, including by developing new methodology when existing theory is inadequate. This role requires statistical reasoning and attributable human and institutional responsibility.
\end{abstract}


\newpage

\section{Introduction}

Artificial intelligence is transforming how data analysis is produced. Natural language systems can generate code, fit large collections of models, propose diagnostics, create visualizations, synthesize research, and draft interpretations. They can also support mathematical derivations, simulation studies, and methodological development. Therefore, any credible assessment of the statistical profession must recognize that many tasks that once required substantial human effort are becoming faster, less costly, or fully automated. The conventional debate distinguishes tasks that AI can perform from those beyond its current capabilities. This framing is inherently unstable because any defense of statistics based on machines being unable to formulate questions, reason causally, develop theory, or exercise judgment may be overtaken by technical progress. More fundamentally, it reduces statistics to a collection of human activities instead of a discipline that establishes the conditions for reliable learning from incomplete information. The relevant question is what conditions must hold for an analytic output to support a scientifically warranted conclusion.

This distinction extends a longstanding statistical tradition. Experimental design connects data to controlled comparisons \citep{fisher1935,cox1958}, while sampling theory connects samples to target populations through an explicit observation scheme \citep{neyman1934}. Decision theory relates information to action through a loss function \citep{wald1950}, and modern causal inference separates observed distributions from targets under specified interventions \citep{rubin1974,dawid2000,pearl2009,hernanrobins2020}. Model criticism and statistical pragmatism further emphasize that inferential conclusions depend on assumptions that require scrutiny \citep{box1976,kass2011,gelmanshalizi2013}. Despite their philosophical differences, none of these traditions regards a numerical output as self-validating.

Recent work on the trustworthiness of statistical inference argues that confidence in a conclusion depends on the trustworthiness of the data, methods, and practices that produce it \citep{hand2022}. Studies of AI-assisted data analysis have examined data scientists' expectations about automation, reviewed the capabilities of LLM-based data agents, and evaluated whether LLMs can select appropriate statistical methods \citep{wang2019,sun2025agents,zhu2024}. This literature shows that AI can participate in multiple stages of analysis, while also documenting limitations in methodological choice, consistency across interactions, and the incorporation of domain knowledge. Related discussions examine the broader relationship among statistics, data science, and AI \citep{friedrich2022,asa2023role,lin2025,donoho2026rebuilding}. A recent essay by \citet{tao2026} applies a similar premise to mathematics. Assuming that AI will become capable of performing a substantial range of research level mathematical tasks, Tao asks which goals and values should guide mathematical research. He distinguishes proof generation from verification, exposition, community evaluation, and incorporation into established theory. A correct proof, on this account, is only one stage in the production of mathematical knowledge.

This article focuses on a different question. In statistical inference, an output may be correctly produced under a given model or procedure without supporting the scientific conclusion for which it is used. Computation alone cannot establish that connection. It depends on a defined target, an observation regime, explicit assumptions, an inferential or decision procedure, an account of uncertainty, and validation under intended use. Following the model of the ASA statement on $p$-values, we present the
argument through a small set of statements supported by explanation and examples \citep{wasserstein2016,wasserstein2019}. 

Drawing on the use of \emph{warrant} in argumentation theory for the reasoning that connects evidence to a claim \citep{toulmin1958,toulmin1984}, we call the structure supporting this connection \emph{statistical warrant}. The concept distinguishes three objects often conflated in analytical practice. An \emph{output} is a number, model, classification, visualization, or narrative. A \emph{procedure} is the rule that maps available information into that output. A \emph{warrant} comprises the reasons and conditions under which the output supports a specified claim or action. Our central thesis is that AI can automate outputs and procedures and assist in constructing statistical warrants, but it cannot eliminate the conditions required for warranted conclusions. As analytical systems generate plausible results at unprecedented speed and scale, statistical work will place greater emphasis on analytical design, methodological innovation, critical evaluation, and responsible oversight.

The remainder of the article is organized as follows. Section 2 defines statistical warrant and uses a nonidentification result to show why computation cannot replace identifying information or assumptions. Section 3 presents five statements concerning targets, the evidential role of design and provenance, distinctions among statistical tasks, analytical abundance, and the role of statistical expertise. Section 4 examines the implications for statistical research, education, institutions, and journals. The final section summarizes the argument and its implications for statistical practice.

\section{Statistical warrant}

This section formalizes \emph{statistical warrant} as the structure connecting observations to a scientific claim or decision. It specifies the target, observation regime, assumptions, procedure, uncertainty, validation, loss, and responsibility required to justify that connection. It then establishes why identification precedes computation and why computational power alone cannot recover a target that is not determined by the observable distribution.

\subsection{A formal representation}

Let $s\in\mathcal{S}$ denote a scientifically admissible state of the system under study. A state specifies the underlying features required to define the scientific target $\tau(s)$ and, under a design or observation regime $d\in\mathcal{D}$, the distribution of the observable data. The regime $d$ specifies the mechanism through which data are collected, measured, and recorded. In a controlled design, this mechanism may involve probability sampling, randomized treatment assignment, or a prespecified measurement schedule. In an observational regime, it may instead reflect convenience sampling, clinical recording practices, nonresponse, censoring, or missingness. For example, in a disease prevalence study using an imperfect diagnostic test, a state may specify the true prevalence, the population composition, and the diagnostic performance. The distribution of the recorded results depends on these features and on the sampling and testing regime.

A state is \textit{scientifically admissible} when it is compatible with substantive knowledge and the assumptions maintained in the analysis, even if it does not describe the true data-generating mechanism. In the diagnostic example, admissibility requires coherent probabilities, defensible measurement assumptions, and externally supported restrictions on diagnostic performance. States violating these conditions are excluded from $\mathcal{S}$, but distinct admissible states may remain compatible with the evidence and even generate the same observable distribution. More generally, $\mathcal{S}$ may include alternative population distributions, causal structures, measurement mechanisms, and deployment conditions. These alternatives may imply different targets, distributions of the observable data, or the relationship between them.

Under a design or observation regime $d\in\mathcal{D}$, the observable random element $O_n$ has \textit{probability law} $\textsf{P}^{d}_{s,n}$. The index $n$ usually denotes the sample size, although it may more generally index the amount or structure of the available information. A probability law specifies the complete distribution of $O_n$. In $\textsf{P}^{d}_{s,n}$, $s$ indexes the underlying scientific state, $d$ the mechanism through which that state is observed, and $n$ the information level. A given state may induce different observable distributions under different regimes, and distinct states may induce the same distribution under a fixed regime.

A substantive question is formalized through a target functional $\tau:\mathcal{S}\rightarrow\Theta$, with $\theta=\tau(s)$, which selects the feature of the scientific state that the analysis seeks to learn, such as a prevalence, causal effect, predictive risk, or optimal decision. Specifying the target determines what constitutes a correct answer and which uncertainty assessment is appropriate. This distinction is essential because the observable distribution depends on both the state and the observation regime, while the target depends on the state. Distinct states may therefore generate the same observable distribution but imply different target values, making the target unidentified without additional assumptions, design information, or external evidence.

A \textit{statistical procedure} is a measurable rule $\delta$ that maps the observed data $O_n$ and any prespecified auxiliary information into an action $a\in\mathcal{A}$, such as an estimate, prediction, interval, test, or decision. Dependence on $n$ is left implicit. Measurability means that $\delta$ is compatible with the probability structures on the data and action spaces, ensuring that $\delta(O_n)$ is a well defined random action and that its probabilities and expected losses exist. 

The adequacy of $\delta$ must be evaluated relative to a criterion that represents the consequences of error. Given a loss function $L$, its risk under state $s$ and regime $d$ is
\[
\Risk_s(\delta)
=
\E^{d}_{s}\big[
L\big(\delta(O_n),\tau(s)\big)
\big],
\]
where $\E^{d}_{s}$ denotes expectation with respect to the probability law $\textsf{P}^{d}_{s,n}$. Thus, the risk is the expected loss over hypothetical realizations of the data generated under state $s$ and regime $d$ \citep{wald1950,blackwellgirshick1954,berger1985}.

A frequentist guarantee may require uniform control of risk over a scientifically relevant subset $\mathcal{S}_0\subseteq\mathcal{S}$ through
\[
\sup_{s\in\mathcal{S}_0}\Risk_s(\delta)
\leq \varepsilon_n,
\]
where $\varepsilon_n$ is an acceptable risk bound that may depend on the information size. This condition ensures that the expected loss does not exceed $\varepsilon_n$ for any state in $\mathcal{S}_0$. The guarantee remains conditional on the specification of $\mathcal{S}_0$, the observation regime, the target, and the loss function, each of which requires scientific justification \citep{wald1950,lecam1986,lehmanncasella1998}.

Under a Bayesian formulation, a prior distribution $\pi(ds)$ on $\mathcal{S}$ induces the posterior distribution $\pi(ds\mid O_n)$. The posterior expected loss of an action $a$ and the corresponding Bayes rule are
\[
\rho_{\pi}(a\mid O_n)
=
\int_{\mathcal{S}}
L\big(a,\tau(s)\big)\,
\pi(ds\mid O_n),
\qquad
\delta_{\pi}(O_n)
\in
\arg\min_{a\in\mathcal{A}}
\rho_{\pi}(a\mid O_n).
\]
Thus, the Bayes rule selects an action by averaging its consequences over posterior uncertainty about the scientific state \citep{savage1954,degroot1970,berger1985}.

The Bayes rule minimizes the posterior expected loss $\rho_{\pi}(a\mid O_n)$ whereas the preceding frequentist criterion controls repeated sampling risk uniformly over $\mathcal{S}_0$. Despite this distinction, both approaches presuppose a scientifically meaningful question, observation regime, target, and loss function. Neither formalism can compensate for an inappropriate specification of these elements. An analysis may therefore be mathematically correct within a model while still lacking statistical warrant.

\begin{definition}[Statistical warrant]
A \emph{statistical warrant} is the structured specification
\[
\W=(Q,\mathcal{S},d,\tau,A,\delta,U,V,L,G).
\]
Here, $Q$ denotes the substantive question, $\mathcal{S}$ the admissible state space, $d$ the design, observation regime, and data provenance, $\tau$ the target, $A$ the identifying and modeling assumptions, and $\delta$ the statistical procedure. The remaining components specify the uncertainty assessment $U$, validation regime $V$, loss function or criterion of scientific adequacy $L$, and governance and accountability structure $G$.
\end{definition}

The components of Definition 1 may be uncertain or contested, but they must be sufficiently explicit to establish what is claimed, why the data are informative about that claim, which assumptions support the connection, how the conclusion may fail, and who is accountable for its use. The procedure $\delta$ may comprise a sequence of analytical operations rather than a single model fit. Records may be linked across sources, missing values imputed, variables extracted from text or images, and labels assigned through automated classification before the final analysis is conducted. These operations depend on assumptions represented in $A$ and may introduce uncertainty that must be addressed in $U$ and $V$. Their relationship to the original records is also part of the data provenance represented by $d$.

This definition is deliberately compatible with likelihood, frequentist, Bayesian, partial identification, and decision theoretic perspectives, which provide different accounts of evidence and uncertainty but all require a defensible connection between the model and the data-generating or observation process \citep{hacking1965,mayocox2006,kass2011,gelmanshalizi2013}. Statistical pluralism does not imply an absence of standards. It requires explicitly justified standards appropriate to different inferential purposes. 

The framework also applies when $\delta$ includes a black-box machine learning method. Interpretability is not a universal condition for supporting a predictive conclusion. Its importance depends on the substantive question, the intended use of the output, and the consequences of error. \citet{lipton2018} distinguishes several purposes of interpretability, including scientific explanation, assessment of individual predictions, detection of undesirable behavior, and justification of decisions. The components of statistical warrant make it possible to state which purpose is relevant and what evidence it requires. \citet{rudin2019} argues that interpretability may be essential when black-box methods are used in high-stakes settings and users must examine the basis of a decision, incorporate substantive information not represented in the data, identify errors, permit affected individuals to challenge a result, or determine when intervention is required. In such settings, predictive performance alone is insufficient. Interpretability may instead form part of the validation criterion $V$, the consequences represented by $L$, or the accountability structure $G$.

\subsection{Identification precedes computation}

The distinction between the observable distribution $\textsf{P}^{d}_{s,n}$ and the scientific target $\tau(s)$ provides the clearest mathematical explanation for why computation alone is insufficient. Define the observation map by $r_d(s)=\textsf{P}^{d}_{s,n}$. The target $\tau$ is identified under $(\mathcal{S},d)$ if
\[
r_d(s_1)=r_d(s_2)
\quad\Longrightarrow\quad
\tau(s_1)=\tau(s_2)
\]
for every $s_1,s_2\in\mathcal{S}$. Equivalently, $\tau$ must be constant across all states that induce the same observable distribution. If this condition fails, observationally indistinguishable states imply different target values, and data generated under regime $d$ cannot determine the target without additional assumptions or information.

\begin{proposition}[Computational power cannot repair nonidentification]
Suppose there exist states $s_1,s_2\in\mathcal{S}$ such that
\[
\textsf{P}^{d}_{s_1,n}
=
\textsf{P}^{d}_{s_2,n}
\,\text{ for every } n,
\qquad
\tau(s_1)\neq\tau(s_2).
\]
Then no sequence of estimators $T_n(O_n)$ can be consistent for $\tau(s)$ at both $s_1$ and $s_2$.
\end{proposition}

\begin{proof}
Let $\theta_k=\tau(s_k)$ for $k=1,2$, and let $\textsf{d}$ be a metric on $\Theta$. Since $\theta_1\neq\theta_2$, it follows that $\Delta=\textsf{d}(\theta_1,\theta_2)>0$. Let $\mathcal{O}_n$ denote the observation space and let $T_n:\mathcal{O}_n\rightarrow\Theta$ be a measurable estimator. Choose $0<\epsilon<\Delta/2$. For each $k=1,2$, define the open $\epsilon$-ball centered at $\theta_k$ by $\mathcal{B}(\theta_k,\epsilon)=\left\{t\in\Theta:\textsf{d}(t,\theta_k)<\epsilon\right\}$, and let
\[
B_{k,n}
=
T_n^{-1}\bigl(
\mathcal{B}(\theta_k,\epsilon)
\bigr)
=
\left\{
o\in\mathcal{O}_n:
\textsf{d}\bigl(T_n(o),\theta_k\bigr)<\epsilon
\right\}.
\]
Thus, $B_{k,n}$ is the event that the estimator computed from the observed data lies within distance $\epsilon$ of $\theta_k$.

The events $B_{1,n}$ and $B_{2,n}$ are disjoint. Otherwise, there would exist some $o\in B_{1,n}\cap B_{2,n}$, and the triangle inequality would imply
\[
\Delta
\leq
\textsf{d}\bigl(\theta_1,T_n(o)\bigr)
+
\textsf{d}\bigl(T_n(o),\theta_2\bigr)
<
2\epsilon
<
\Delta,
\]
which is impossible. For every $n$, the equality $\textsf{P}^{d}_{s_1,n}=\textsf{P}^{d}_{s_2,n}$ allows us to denote the common observation law by $\textsf{P}_n$. If the sequence $\{T_n\}$ were consistent at both states, then $\textsf{P}_n(B_{1,n})\rightarrow 1$ and $\textsf{P}_n(B_{2,n})\rightarrow 1$. However, because $B_{1,n}$ and $B_{2,n}$ are disjoint, $\textsf{P}_n(B_{1,n})+\textsf{P}_n(B_{2,n})=\textsf{P}_n(B_{1,n}\cup B_{2,n})\leq 1$ for every $n$. Taking limits would imply $2\leq 1$, which is impossible. Therefore, no sequence of estimators based on $O_n$ can be consistent for $\tau(s)$ at both $s_1$ and $s_2$.
\end{proof}

Although elementary, the proposition has a fundamental implication. No amount of additional data from the same observation regime, model complexity, computational search, or AI capability can identify a target that the observable law does not distinguish. Resolving nonidentification requires new design information, stronger assumptions, external evidence, a reformulated target, or reporting an identified set \citep{manski2003,greenland1999,pearl2009}. AI can help formulate and evaluate these alternatives, but it cannot choose among them using the observable distribution alone.


\begin{principle}[Epistemic non-substitution]
Greater computational power can reduce numerical, optimization, and search errors within a well-specified statistical problem. By itself, however, it cannot define the scientific target, validate the observation regime, justify untestable assumptions, determine which consequences matter, or assign responsibility for a decision.
\end{principle}

This principle does not posit an intrinsic boundary between human and machine intelligence. Future AI systems may contribute to every component of $\W$, but those components remain necessary. The principle also distinguishes epistemic competence from institutional responsibility. Even a highly capable system cannot by itself confer authority to define a loss function, accept residual risk, or assume accountability for harm.

\section{Five statements}

The framework of statistical warrant leads to five statements on the role of statistics in AI-mediated analysis. They concern target specification, the evidential basis of data, the distinction among inferential tasks, the consequences of analytical abundance, and the responsibilities of the statistician. These statements apply to any analytical system whose outputs are intended to support scientifically warranted conclusions, regardless of current AI capabilities.

\statementheading{The question and target are part of the method}{A statistical analysis is not defined by an algorithm and a dataset. It is defined first by the question, target population, target functional, and inferential or decision objective.}

An algorithm may be technically valid yet scientifically irrelevant. The instruction ``predict $Y$ from $X$'' leaves unspecified the target population, prediction horizon, intervention regime, loss function, and intended use. Likewise, ``estimate the treatment effect'' does not define the treatment versions, comparison, population, time point, or causal contrast. These specifications determine the mathematical target and the criterion by which the analysis must be evaluated \citep{shmueli2010,hernan2019,efron2020,vanderlaanrose2011}.

Formally, neither the data nor the algorithm determines the target. The target is defined by translating a substantive question $Q$ into a functional $\tau:\mathcal{S}\rightarrow\Theta$ on the admissible state space $\mathcal{S}$. Different questions may therefore define different estimands from the same observed records. An average treatment effect, an effect among the treated, and a policy value are distinct functionals, just as a predictive mean, a conditional quantile, and a classification rule optimized under asymmetric costs are distinct targets. Target specification is therefore an integral part of the statistical method.

AI heightens this requirement because it can produce polished answers before the underlying question is precisely formulated. Such fluency may conceal \emph{target substitution}, whereby a difficult scientific question is replaced by one that is easier to compute but substantively different. For example, when asked whether a treatment reduces mortality, an AI system may fit a predictive model and report a treatment coefficient or variable importance instead of a causal effect. A primary role of the statistician is to prevent this substitution by translating the substantive problem into a scientifically meaningful and mathematically well defined target.

\statementheading{Data acquire evidential meaning through design, provenance, and assumptions}{The size or complexity of a dataset does not determine its evidential value. Evidence depends on how observations, variables, labels, and missing values were generated and on the assumptions connecting that process to the target.}

The same numerical values may support different conclusions depending on whether they arise from probability sampling, sampling conditional on the outcome, convenience sampling, a controlled experiment, or an observational study \citep{cochran1977,prenticepyke1979,sarndal1992}. The design determines the repetitions, comparisons, or counterfactual regimes that give uncertainty its meaning \citep{fisher1935,neyman1934,cox1958}. Data provenance identifies the population and measurement process represented by the records, while assumptions determine which conclusions may extend beyond a description of the observed data. 

Missing data provide a simple illustration. Let $Y$ be a binary outcome and let $R=1$ indicate that $Y$ is observed. By the law of total probability, we have that
\[
\Prb(Y=1)
=
\Prb(Y=1\mid R=1)\Prb(R=1)
+
\Prb(Y=1\mid R=0)\Prb(R=0).
\]
The observed data identify $\Prb(Y=1\mid R=1)$ and $\Prb(R=1)$, but not the outcome distribution among nonrespondents, $\Prb(Y=1\mid R=0)$. Hence, identification requires additional information such as a missing at random assumption conditional on observed covariates, auxiliary data, a sensitivity model, or partial identification bounds, each of which may yield different conclusions \citep{littlerubin2019,manski2003}. No imputation algorithm can make the unidentified component empirically observed.

Moreover, increasing sample size does not resolve these problems. Selection bias may dominate sampling variability, causing a large but nonrepresentative dataset to produce highly precise yet misleading conclusions \citep{meng2018}. Label construction, measurement error, undocumented exclusions, and preprocessing choices may likewise determine what a model learns. Documentation tools such as datasheets and model cards are valuable because they make data provenance and intended use explicit, although documentation alone cannot establish validity \citep{gebru2021,mitchell2019}.

AI can support data validation, anomaly detection, pipeline reconstruction, and sensitivity analysis. However, data provenance can be compromised when generated transformations or default settings are not documented. For example, an automated pipeline may discard records with missing covariates, silently changing the analyzed population and its relationship to the target population. The statistical obligation remains unchanged. Assumptions must be explicit, their empirical implications examined whenever possible, and conclusions restricted to what the design and data provenance support.

\statementheading{Inferential tasks are mathematically distinct}{Description, prediction, explanation, causal inference, and decision are not interchangeable labels. They have different targets, assumptions, losses, and validation criteria. Success at one task does not establish success at another.}

Consider prediction, causal inference, and decision. Let $\textsf{P}_{\star}$ denote the deployment distribution and $\E_{\star}$ expectation with respect to it. A predictive rule may target
\[
f^\star
\in
\arg\min_{f\in\mathcal{F}}
\E_{\star}
\left[
\ell\bigl(Y,f(X)\bigr)
\right],
\]
where $X$ is the vector of predictors, $Y$ is the outcome, $\mathcal{F}$ is the class of candidate prediction functions, and $f(X)$ is the prediction produced from $X$. The loss $\ell(Y,f(X))$ measures predictive error, while the expectation averages this loss over the joint deployment distribution of $(X,Y)$. A causal analysis may instead target the average treatment effect
\[
\tau_{\textsf{ATE}}
=
\E_{\mathrm{tar}}
\left[
Y(1)-Y(0)
\right],
\]
where $Y(1)$ and $Y(0)$ are the potential outcomes under treatment and control, respectively, and $\E_{\mathrm{tar}}$ denotes expectation over the target population. This quantity is not generally identified from the observational distribution without an appropriate design or additional identifying assumptions. A decision rule may target
\[
\delta^\star
\in
\arg\min_{\delta}
\E_{\star}
\left[
L\bigl(\delta(X),S\bigr)
\right],
\]
where $\delta(X)\in\mathcal{A}$ is the action selected after observing $X$, $S$ is the state that determines the consequences of that action, and $L(\delta(X),S)$ quantifies the associated loss when the state is $S$. In this case, $\E_{\star}$ averages loss over the joint deployment distribution of $(X,S)$. These objectives may involve the same data and algorithms, but they define mathematically distinct problems \citep{wald1950,shmueli2010,pearl2009,hernanrobins2020}.

High predictive accuracy does not identify a causal effect, and causal identification does not ensure useful individual prediction. Even calibrated probabilities do not prescribe an action until costs, constraints, and available interventions are specified. Proper scoring rules can align forecast evaluation with a declared predictive objective, but they cannot determine which objective is substantively appropriate \citep{gneiting2007}. Likewise, the incompatibility of common fairness criteria under general data conditions shows that predictive performance alone cannot determine a unique normative objective \citep{chouldechova2017,selbst2019}.

AI systems are particularly vulnerable to task confusion because the same interface can answer distinct formulations with similar confidence and style. For example, a system asked who should receive a treatment may recommend individuals with the highest predicted risk, although outcome risk and individual treatment benefit are different targets. The statistician must distinguish association from intervention, estimation from prediction, and prediction from action. Confusing these tasks changes the mathematical object under study.

\statementheading{Analytical abundance increases the need for statistical control}{AI lowers the cost of generating and revising analyses more rapidly than it lowers the cost of validating them. The resulting analytical abundance amplifies adaptive selection, hidden multiplicity, model uncertainty, and deployment risk.}

Selection is adaptive when decisions about which analyses to conduct, modify, or report depend on results already observed in the same data. Examples include adding variables after inspecting model diagnostics, redefining an outcome after examining preliminary results, and selecting a model after repeated evaluation on the same validation set. Multiplicity may remain hidden when the reported analysis does not disclose the full set or sequence of analyses considered.

Even when each candidate analysis is valid in isolation, selecting which result to report changes its statistical properties. Consider first a fixed collection of tests. Suppose $m$ true null hypotheses produce independent $p$-values uniformly distributed on $[0,1]$, and only the smallest is reported. Then the probability of observing at least one nominally significant result at level $\alpha$ is
\[
\Prb\left(
\min_{1\leq j\leq m}p_j\leq\alpha
\right)
=
1-(1-\alpha)^m.
\]
For $m=100$ and $\alpha=0.05$, this probability is approximately $0.994$. Although independence yields the simple expression above, the central issue is adaptive selection. AI can create a more complex problem through automated feature construction, prompt revision, model search, outcome redefinition, and repeated inspection of the same validation data.

More generally, let $Z$ denote the data, let $\widehat{\mu}_j(Z)$ be the statistic produced by candidate analysis $j$, and define $\mu_j=\E(\widehat{\mu}_j(Z))$. Let $J$ denote the data dependent index selected through an adaptive analytical transcript. If each centered statistic $\widehat{\mu}_j-\mu_j$ is $\sigma$-sub-Gaussian, then, under standard regularity conditions, the expected selection bias satisfies
\[
\left|
\E\bigl[
\widehat{\mu}_J-\mu_J
\bigr]
\right|
\leq
\sigma\sqrt{2\,I(J;Z)},
\]
where $I(J;Z)$ is the mutual information between the selected analysis and the data. This quantity measures how much information the selected analysis $J$ contains about the observed data $Z$ and quantifies the data dependence induced by the selection process \citep{russozou2020}. The bound shows that adaptive selection can create bias even when every candidate statistic is well behaved when considered in isolation. Selective inference, reusable holdouts, preregistration, nested validation, and independent replication address different aspects of this problem \citep{leebpotscher2005,dwork2015,taylortibshirani2015,gelmanloken2014}.

Uncertainty quantification must extend beyond the output conditional on a selected model. A posterior interval or standard error may represent within-model variation while omitting uncertainty from measurement error, missingness, model selection, model discrepancy, computational approximation, or distribution shift \citep{chatfield1995,draper1995,kennedyohagan2001,hullermeier2021}. For a deployed rule $\delta$, the relevant risk is
\[
\Risk_{\star}(\delta)
=
\E_{\star}
\left[
L\bigl(\delta(X),Y\bigr)
\right],
\]
where $\E_{\star}$ denotes expectation under the deployment distribution $\textsf{P}_{\star}$. Risk estimated under a validation distribution $\textsf{P}_{\mathrm{val}}$ is not generally informative about $\Risk_{\star}(\delta)$ without transport assumptions relating $\textsf{P}_{\mathrm{val}}$ to $\textsf{P}_{\star}$. Estimating $\Risk_{\star}(\delta)$ from validation data requires a
defensible relationship between the validation and deployment distributions. This may involve assuming that the two distributions are equal or that relevant conditional distributions remain stable, and that the deployment population is adequately
represented in the validation data \citep{quinonero2009}. Studies using newly collected test sets and controlled distribution shifts have documented declines in predictive accuracy and in the reliability of uncertainty estimates when these conditions fail \citep{recht2019,ovadia2019}.


These examples identify two ways in which AI-assisted analysis can weaken statistical warrant. Before a result is reported, the same data may be used repeatedly to generate, revise, and select analyses. After a model is deployed, the population or measurement process may differ from the conditions under which it was validated. For instance, consider a clinical triage model selected after repeated evaluation on the same validation data. Its reported performance may already reflect selection. If patient composition or recording practices later change, the original validation results may also cease to represent deployment performance. The model may then require independent evaluation, recalibration, or suspension. Audit and risk-management frameworks can document the selection process, validation conditions, monitoring criteria, and rules for intervention \citep{raji2020,nist2023,nist2024}. These procedures contribute to statistical warrant only when the target, assumptions, uncertainty assessment, and intended use have been appropriately specified.


\statementheading{The statistician constructs and safeguards statistical warrant}{The fundamental role of the statistician is not the mechanical application of a fixed catalogue of methods. It is to formulate, create, evaluate, and govern the inferential structures that make learning from data scientifically defensible.}

This statement captures both continuity and transformation. As routine execution becomes increasingly automated, the statistical function becomes more important because both analyses and persuasive but invalid conclusions can be generated at unprecedented speed and scale. Statistical expertise therefore shifts from operating procedures toward designing, evaluating, and criticizing the inferential systems that give those procedures scientific meaning. Table~\ref{tab:roles} organizes this work into four related functions: defining the analytical problem, developing appropriate methods, evaluating the resulting evidence, and overseeing its interpretation and use.

\begin{table}[H]
\centering
\footnotesize
\setlength{\tabcolsep}{4pt}
\renewcommand{\arraystretch}{1.15}

\begin{tabularx}{\textwidth}{
    >{\hsize=0.75\hsize\bfseries\raggedright\arraybackslash}X
    >{\hsize=1.00\hsize\raggedright\arraybackslash}X
    >{\hsize=1.25\hsize\raggedright\arraybackslash}X
}
\toprule
Function & Primary object & Central question \\
\midrule
Architect
& Question, target, design, and loss
& What must be learned, from whom, under which observation regime, and for what use? \\

Methodologist
& Model, assumptions, theory, and computation
& What is genuinely new about the problem, and what procedure has defensible properties? \\

Evaluator and auditor
& Uncertainty, validation, adaptivity, and shift
& Which guarantees hold, what has been selected, and where can the system fail? \\

Steward of evidence
& Interpretation, communication, governance, and record
& What may be concluded, with what limitations, and who is responsible for acting on it? \\
\bottomrule
\end{tabularx}

\caption{Core functions of the statistician in AI-mediated analysis.}
\label{tab:roles}
\end{table}

These functions overlap in practice and may be performed by different people at different stages of an analysis. In particular,  methodological development becomes necessary when existing procedures do not provide adequate warrant for a new target, observation regime, or form of data use. AI-assisted analysis creates statistical problems that are not covered by theory developed for a prespecified procedure applied once to a fixed dataset. When variables, models, outcomes, or prompts are revised in response to intermediate results, uncertainty conditional on the final analysis does not represent the full selection process. When a deployed system influences subsequent decisions or data collection, assumptions of independence or stability may also fail. New methods are needed to define targets under these conditions, determine when they are identified, and provide uncertainty assessments with stated operating properties. This view continues the broader tradition of treating statistical methodology as a response to changes in how data are generated and analyzed \citep{tukey1962,chambers1993,breiman2001,donoho2017}.

AI can assist with many parts of statistical work, including literature review, derivations, model formulation, programming, and sensitivity analysis. Technical correctness, however, does not establish that the resulting analysis is appropriate for the scientific question. For instance, consider a longitudinal analysis in which an AI system proposes and implements a model under the assumption that missingness due to dropout is ignorable. The implementation may be correct, but that does not establish that the assumption is appropriate. If dropout depends on unobserved outcomes after conditioning on the observed data, the full data distribution is not identified without additional assumptions. Constructing the statistical warrant requires stating this assumption, explaining its role in identification, and examining sensitivity to plausible departures from it. These judgments belong to $\W$. AI may support them, but they must remain open to criticism and embedded in accountable scientific practice.

Importantly, the statement concerns the functions described in Table~\ref{tab:roles} not professional designation. Individuals without the title \emph{statistician} may perform these functions, while statisticians may fail to do so. Future AI systems may also exhibit substantial statistical reasoning. What remains indispensable is the statistical function itself. In consequential settings, accountability cannot be delegated to a model but must remain attributable to people and institutions with the competence, authority, information, and resources required to intervene \citep{asa2022ethics,asa2024ethicalai,raji2020}.

\section{Consequences for statistical science and practice}

The preceding statements have direct implications for how statistical knowledge is developed, taught, evaluated, and governed. This section translates statistical warrant into priorities for research, education, and institutional practice in AI-mediated analysis.

\subsection{Research}

The central research challenge extends beyond constructing high-performing models to maintaining statistical warrant when AI contributes to problem formulation, data processing, model selection, interpretation, and decisions. Priority areas include inference after repeated data dependent exploration by analysts and AI systems, propagation of uncertainty through multistage analyses, evaluation when the population or measurement process changes between model development and use, causal inference when decisions based on a model affect subsequently observed data, assessment of the factual and statistical reliability of generated conclusions, and decision procedures that defer to human review when available information is insufficient. These problems require computational advances, but their central questions concern the target, the process that generated the available data, the assumptions connecting those data to the intended use, and the consequences of error \citep{yu2020,lin2025,donoho2026rebuilding}.

Methodological research should report more than asymptotic guarantees under an idealized model or performance on a familiar benchmark. It should identify the target and observation regime for which a method is intended, state the assumptions supporting its guarantees, describe the relevant sources of uncertainty, and assess sensitivity to plausible violations. It should also distinguish error arising from computational approximation from error due to measurement, sampling, model misspecification, or a mismatch between the stated target and the scientific question. Stability across defensible analytical choices can show that a conclusion is not driven by a particular specification, but stability alone cannot establish identification or scientific relevance \citep{yu2020}.

Human-AI workflows present an additional problem. Humans and AI systems may jointly formulate questions, construct analytical data through operations such as record linkage and variable extraction, select models, and interpret results. Selective inference and adaptive data analysis address particular forms of data dependent choice, but provide only a partial basis for inference when the question, target, analytical data, and procedure may all be revised during the interaction \citep{dwork2015,taylortibshirani2015}. The need for methods to evaluate and validate multistep human-AI collaboration is also reflected in current research priorities \citep{nsf2020humanai}.

\subsection{Education}

The value of AI assistance depends partly on the knowledge of the user. For experienced statisticians, AI may complement expertise developed through formal training and practice by assisting with computation, programming, and exploration. For students who have not yet developed that expertise, the same assistance may conceal errors or conceptual gaps that they are not prepared to recognize. Evidence that humans and LLMs make different types of errors when selecting statistical methods suggests possible complementarity, although effective collaboration cannot be assumed \citep{zhu2024}. Statistical training requires opportunities to reason and work without AI alongside supervised experience using it critically.

These considerations shift the emphasis of statistical education beyond routine computation without eliminating it. Students continue to require probability, inference, experimental and sampling design, causal reasoning, decision theory, optimization, and scientific computing. They must also learn to check generated code and derivations, detect target substitution, analyze adaptive workflows, evaluate uncertainty when populations or measurement processes change, and communicate limitations. Assessment can emphasize constructing and criticizing \W. Documentation is also part of this training. When AI contributes to data preparation, model selection, or interpretation, students can record the system used, prompts and outputs, generated code and revisions, analytical alternatives considered, and decisions made by the analyst. Such records help reconstruct the analytical process, evaluate data dependent selection, and trace responsibility for the final conclusion.

Faculty preparation is part of the same problem. Higher education institutions are increasingly encouraging faculty to develop relevant AI skills \citep{robert2026}, yet many instructors lack confidence or adequate preparation for using AI in teaching \citep{lee2024,ruediger2024}. Statisticians who teach need sufficient familiarity with AI to evaluate its outputs, design assignments that preserve statistical learning, establish clear expectations for disclosure, and prepare students for analytical work in which AI is routinely available.

\subsection{Institutions and journals}

One way to make statistical warrant reviewable is through a concise inferential record for consequential analyses. Depending on the purpose of the analysis, the record may identify

\begin{enumerate}[label=\arabic*.]
  \item The substantive question, target population, and estimand or decision.
  \item The origin, selection, measurement, and transformation of the data.
  \item The identifying and modeling assumptions.
  \item The adaptive analytical path that materially affected the reported result.
  \item The sources of uncertainty represented and omitted.
  \item The validation population, loss function, calibration criterion, and stress conditions.
  \item The conditions for deployment, monitoring, escalation, suspension, and retirement.
  \item The responsible individuals and institutions, including the role of AI assistance.
\end{enumerate}

This record operationalizes $\W$ without replacing scientific judgment. Journals can promote it through protocol registration, reproducible materials, disclosure of AI assistance, sensitivity analysis, and explicit separation of exploratory, confirmatory, predictive, causal, and decision claims \citep{nasem2019,pineau2021}. Institutions can support the same objective through independent validation, audit access, and procedures for reconsidering a system when its assumptions or performance no longer support its intended use \citep{sculley2015,raji2020,nist2023}.

\section{Conclusion}

AI transforms the production of statistical analysis but not the logical conditions of valid learning from data. Observations may fail to identify scientific targets, uncertainty remains conditional on an inferential structure, distinct tasks require distinct criteria, and decisions involve consequences not contained in the data. These are not temporary limitations of current AI systems but structural features of reasoning and acting under incomplete information.

The strongest case for statistics in the age of AI therefore does not rest on uniquely human intuition or on the assumption that machines will never reason statistically. It rests on the fact that scientific claims require statistical warrant. Targets must be defined, observation regimes and provenance understood, assumptions justified, uncertainty and validation aligned with intended use, adaptive selection incorporated into inference, and losses and responsibilities made explicit.

The statistician remains fundamental because the discipline is organized around constructing, scrutinizing, and safeguarding these elements. As routine execution becomes automated, the profession will increasingly emphasize analytical architecture, methodological innovation, evaluation, audit, and stewardship. Its most valuable contribution will be to ensure that analyses, however produced, deserve to inform science and action.

\section*{Generative-AI assistance disclosure}

Generative AI was used to assist drafting, structural editing, and LaTeX preparation. The authors are responsible for verifying the sources, arguments, mathematical statements, and final wording before submission. No confidential data were used in preparing this manuscript.





\thispagestyle{fancy}
\bibliographystyle{apalike}
\bibliography{references.bib}

\end{document}